\documentclass[11pt,a4paper]{article}

\usepackage[T1]{fontenc}
\usepackage[utf8]{inputenc}
\usepackage{lmodern}
\usepackage[margin=2.7cm]{geometry}
\usepackage{amsmath,amssymb,amsthm,mathtools}
\usepackage{bm}
\usepackage{booktabs}
\usepackage{enumitem}
\usepackage[dvipsnames]{xcolor}
\usepackage{authblk}
\usepackage{hyperref}
\hypersetup{
  colorlinks=true,
  linkcolor=MidnightBlue,
  citecolor=MidnightBlue,
  urlcolor=MidnightBlue
}

\numberwithin{equation}{section}

\theoremstyle{definition}
\newtheorem{definition}{Definition}[section]
\newtheorem{example}[definition]{Example}
\newtheorem{remark}[definition]{Remark}
\theoremstyle{definition}
\newtheorem{theorem}[definition]{Theorem}
\newtheorem{proposition}[definition]{Proposition}
\newtheorem{corollary}[definition]{Corollary}
\newtheorem{lemma}[definition]{Lemma}
\newtheorem{problem}[definition]{Problem}
\usepackage[capitalize,noabbrev]{cleveref}

\newcommand{\F}{\mathbb{F}}
\newcommand{\PG}{\mathrm{PG}}
\newcommand{\E}{\mathbb{E}}
\newcommand{\Span}{\langle}
\newcommand{\cG}{\mathcal{G}}
\newcommand{\cC}{\mathcal{C}}
\newcommand{\cF}{\mathcal{F}}
\newcommand{\supp}{\mathrm{supp}}

\newcommand{\cI}{\mathcal{I}}
\newcommand{\cD}{\mathcal{D}}
\newcommand{\Tmax}{T_{\max}}
\newcommand{\Tavg}{T_{\mathrm{avg}}}
\newcommand{\rk}{\operatorname{rk}}
\newcommand{\Aut}{\operatorname{Aut}}
\newcommand{\qbinom}[2]{\genfrac{[}{]}{0pt}{}{#1}{#2}_q}

\title{The Generalized Random Access Problem for Linear Codes}

\author[1]{Anina Gruica}
\affil[1]{Technical University of Denmark, Denmark}
\author[2]{Antonio Petrillo}
\author[2]{Ferdinando Zullo}
\affil[2]{Universit\`a degli Studi della Campania ``Luigi Vanvitelli'', Italy}
\date{}

\begin{document}
\maketitle

\begin{abstract}
Random access is a central requirement in DNA-based storage systems: one would like to recover selected information symbols without sequencing the whole encoded object.  A recent combinatorial model associates to a generator matrix $G\in \F_q^{k\times n}$ the random variable $\tau_i(G)$, measuring the number of sampled columns needed to recover the information vector $e_i$.  {We study the cardinality-based extremal and finite-geometric aspects of simultaneous multi-symbol recovery.  For a nonempty set $I\subseteq[k]$, let $\tau_I(G)$ denote the number of random column samples needed until all vectors $e_i$, $i\in I$, lie in the span of the observed columns.  This variable interpolates between the singleton random access problem and the full-recovery problem underlying coverage depth.  For each $m$, we introduce uniform worst-case and average parameters over all requested sets $I$ with $|I|=m$.  Using the known subset-counting formula for $\E[\tau_I(G)]$, we establish general upper and lower bounds for these parameters.  In particular, the lower bounds are expressed through order statistics of the singleton recovery variables and specialize to the known singleton bounds when $m=1$.  For systematic MDS encoders, we record an equivalent form of the known multi-symbol expectation formula and derive monotonicity and asymptotic consequences.  For simplex encoders in arbitrary dimension, we obtain closed formulae in terms of Gaussian binomial coefficients; the full-recovery endpoint agrees with the known coverage-depth formula for simplex codes.  Finally, in dimension three we study balanced quasi-arcs and compare their values with the simplex and MDS benchmarks.  The resulting formulae exhibit three regimes: balanced quasi-arcs improve the recovery of one fundamental information symbol, remain slightly better than systematic MDS encoders for the recovery of two fundamental symbols in length-matched examples, whereas systematic MDS encoders attain the optimal full-recovery value whenever they exist.}
\end{abstract}

\medskip
\noindent\textbf{Keywords.} DNA storage; random access; linear codes; MDS codes; simplex codes; finite projective geometry; balanced quasi-arcs.

\medskip
\noindent\textbf{MSC 2020.} 94B05, 94B27, 51E20, 60C05.

\section{Introduction}

DNA-based data storage has motivated a wide range of coding-theoretic questions in which information is represented by large unordered collections of short DNA strands and sequencing is modeled as a random sampling process. Besides reliability and storage density, an important goal is \emph{random access}: the ability to retrieve a prescribed part of the stored information without sequencing the entire encoded object. Random access has been demonstrated experimentally through selective amplification techniques, while its coding-theoretic and information-theoretic aspects have been investigated in several recent works; see, for instance, \cite{organick2018random,yazdi2015dna,shomorony2022information,milenkovic2024review}.

The model considered in this paper starts from a linear code with generator matrix $G\in\F_q^{k\times n}$. The columns of $G$ represent the encoded symbols that are sampled independently and uniformly at random, with replacement. Recovering the $i$-th information component is equivalent to observing enough columns so that the standard basis vector $e_i\in\F_q^k$ belongs to their linear span. The corresponding singleton random access problem was introduced in the coverage-depth framework of \cite{bar2024cover} and subsequently studied using combinatorial, geometric, probabilistic, and algorithmic methods in \cite{gruica2024combinatorial,gruica2024geometry,boruchovsky2025making,bodur2025randomvariables,wang2026randomaccess}. In particular, these works investigate exact expectation formulae, extremal constructions, geometric configurations, lower bounds, efficient computational methods, and the full distribution of the associated recovery variables.

At the opposite extreme, one may ask to recover all the basis vectors $e_1,\ldots,e_k$, or equivalently the entire information vector. This is the full-recovery coverage-depth problem, for which MDS codes attain the optimal expected retrieval time; see \cite{bar2024cover,bertuzzo2025coverage,bertuzzo2026duality}. The singleton and full-recovery problems therefore describe two extremal regimes: the recovery of one prescribed information symbol and the recovery of the complete information object.

The purpose of this paper is to develop a uniform framework for the intermediate regime. Given a nonempty subset $I\subseteq[k]$, we study the number of random column samples needed, on average, to recover all information vectors $e_i$ with $i\in I$. Thus, the case $|I|=1$ gives the classical singleton random access problem, whereas $I=[k]$ gives full recovery. The values $1<|I|<k$ describe the transition between these two regimes and allow one to investigate which properties of an encoder favor partial recovery and which favor global recovery.

A closely related block-structured retrieval problem has recently been studied by Bar-Lev in \cite{BarLev2026CodedRetrieval}. In that setting, the information symbols are partitioned into two complementary files, and the main objective is to determine the achievable trade-off between their expected retrieval times. By contrast, our approach is cardinality-based: for each $m$, we consider all requested subsets $I\subseteq[k]$ with $|I|=m$ and study both their worst-case and average recovery behavior. This leads naturally to the parameters $\Tmax^{(m)}$ and $\Tavg^{(m)}$, which provide a uniform sequence of extremal quantities interpolating between singleton random access and full recovery.

We use the general subset-counting formula for $\E[\tau_I(G)]$, already implicit in the framework of \cite{gruica2024combinatorial} and also employed in \cite{BarLev2026CodedRetrieval}. This formula expresses the expectation in terms of the numbers $\alpha_I(s)$ of $s$-subsets of columns whose span contains all the requested vectors. The problem is therefore reduced to a finite-geometric counting question. Building on this framework, we establish general upper and lower bounds for $\Tmax^{(m)}$ and $\Tavg^{(m)}$. In particular, our lower bounds are obtained through the order statistics of the singleton recovery variables and specialize, for $m=1$, to the known singleton bounds from \cite{bar2024cover}. For systematic encoders, we also develop a dual-code interpretation showing that nontrivial recovering subsets are controlled by low-weight dual codewords involving the requested systematic positions.

Symmetry plays an important role in this framework. If the information-preserving automorphism group of the encoder acts transitively on the requested information sets of a fixed cardinality, then $\E[\tau_I(G)]$ depends only on $|I|$ and not on the particular choice of $I$. More generally, the computation of $\Tmax^{(m)}$ and $\Tavg^{(m)}$ can be reduced to representatives of the corresponding automorphism-group orbits.

We then evaluate the generalized recovery expectations for several important families. For systematic MDS encoders, we recover an equivalent form of the expectation formula appearing in \cite{BarLev2026CodedRetrieval}, expressed explicitly as a function of $m=|I|$, and derive monotonicity and asymptotic consequences. We also analyze simplex encoders in arbitrary dimension. Their projective system is the whole space $\PG(k-1,q)$, and the computation reduces to counting subsets of projective points whose span contains a prescribed coordinate subspace. The full-recovery endpoint agrees with the simplex coverage-depth formula obtained in \cite{bertuzzo2026duality}.

Finally, in dimension $k=3$, we study balanced quasi-arcs, which were introduced in the random access setting in \cite{gruica2024geometry} and further analyzed in \cite{bodur2025randomvariables}. These configurations privilege three fundamental points and their joining lines, making them particularly suitable for investigating partial recovery. We compare their generalized recovery values with the simplex and systematic MDS benchmarks. The resulting formulae exhibit three different regimes: balanced quasi-arcs improve the recovery of one fundamental information symbol, remain slightly better than systematic MDS encoders for the recovery of two fundamental symbols in length-matched examples, whereas systematic MDS encoders attain the optimal full-recovery value whenever they exist.

\paragraph{Organization.}
\Cref{sec:model} introduces the generalized random access variable, the counting function $\alpha_I(s)$, and the extremal parameters $\Tmax^{(m)}$ and $\Tavg^{(m)}$. {\Cref{sec:general-expectation} recalls the general expectation formula, includes a proof for completeness, records basic properties, and discusses the dual-code viewpoint.}  \Cref{sec:bounds} gives general upper and lower bounds.  \Cref{sec:mds} treats systematic MDS encoders.  \Cref{sec:simplex-general} treats simplex encoders in arbitrary dimension.  \Cref{sec:dimension-three} contains the dimension-three analysis for balanced quasi-arcs, together with numerical comparisons.  We conclude with a discussion and some directions for future work.

\section{The generalized random access problem}\label{sec:model}

Let $G\in \F_q^{k\times n}$ be a generator matrix of a linear code, and denote its columns by
\[
g_1,\ldots,g_n\in \F_q^k.
\]
The columns of $G$ represent the encoded symbols that can be queried in order to recover information symbols.  In the classical random access problem, one fixes an index $i\in [k]$ and studies how many randomly chosen encoded symbols are needed, on average, to recover the $i$-th information symbol.  In linear-algebraic terms, this amounts to asking when the canonical vector $e_i$ belongs to the span of the columns that have been observed.

In this paper we consider a simultaneous version of the same problem.  Instead of recovering a single information symbol, we fix a nonempty set
\[
I\subseteq [k]
\]
and ask for the number of random queries needed to recover all information symbols indexed by $I$.  Equivalently, we want the whole set of canonical vectors
\[
\{e_i : i\in I\}
\]
to be contained in the span of the columns collected so far.

{The reason for considering this intermediate problem is that it connects two extremal situations that have previously been studied separately.}  When $|I|=1$, one recovers the usual random access problem for one information coordinate.  When $I=[k]$, the problem becomes the recovery of the full information vector, and it is closely related to the coverage depth problem.  Thus the generalized random access problem interpolates between local recovery and global recovery.

Studying the range
\[
1<|I|<k
\]
allows us to understand how the expected number of random accesses changes as the amount of requested information increases.  In this sense, the generalized problem provides a bridge between the local behavior of a code, measured by the recovery of one coordinate, and its global behavior, measured by the recovery of all coordinates.  This intermediate point of view can help identify which structural properties of the generator matrix are responsible for good performance in the extremal cases.

We assume that the columns are sampled independently, uniformly at random, and with replacement.  Thus, if $X_1,X_2,\ldots$ is a sequence of independent random variables uniformly distributed on $[n]$, then after $r$ queries the set of distinct columns that have been observed is
\[
S_r=\{X_1,\ldots,X_r\}\subseteq [n].
\]
Notice that the number of queries is $r$, while $|S_r|$ may be smaller than $r$, since repetitions are allowed in the sampling process.

\begin{definition}[Generalized random access variable]
Let $\emptyset\neq I\subseteq [k]$.  The generalized random access variable associated with $I$ is the random variable
\[
\tau_I(G):=
\min\left\{
r\geq 1:
\{e_i : i\in I\}\subseteq
\Span\{g_j : j\in S_r\}\rangle
\right\}.
\]
When $I=\{i\}$, we write $\tau_i(G)$, recovering the usual random access variable for the $i$-th information symbol.
\end{definition}

\begin{remark}\label{rem:encoder-dependence}
The generalized random access variable $\tau_I(G)$ depends on the chosen
encoder $G$, and not only on the abstract linear code generated by $G$.
Indeed, replacing $G$ by $MG$, with $M\in \mathrm{GL}_k(q)$, gives a
generator matrix for the same code, but changes the interpretation of the
information coordinates $e_1,\ldots,e_k$.  Thus, throughout the paper, we
regard $G$ as part of the data.  When we speak about systematic MDS
encoders, we always mean systematic generator matrices $G=(I_k\mid A)$.
\end{remark}

The quantity $\tau_I(G)$ measures the efficiency of the code for the simultaneous recovery of the information symbols indexed by $I$.  Small values of $\E[\tau_I(G)]$ mean that, on average, few random accesses are sufficient to reconstruct the requested information symbols.

For a fixed integer $m$, it will be useful to collect all requested subsets of size $m$ in the notation
\[
\cI_m:=\{I\subseteq [k]: |I|=m\}.
\]
We define the corresponding maximum and average expected numbers of random accesses by
\[
\Tmax^{(m)}(G):=
\max_{I\in \cI_m}\E[\tau_I(G)]
\]
and
\[
\Tavg^{(m)}(G):=
\frac{1}{\binom{k}{m}}
\sum_{I\in \cI_m}\E[\tau_I(G)].
\]
For $m=1$ these are exactly the parameters usually denoted by $\Tmax$ and $\Tavg$ in the singleton random-access problem.  Thus $\Tmax^{(m)}(G)$ and $\Tavg^{(m)}(G)$ provide a uniform notation for all levels of recovery, from one information symbol to the full information vector.

\begin{problem}[Generalized random access problem]\label{prob:generalized-random-access}
Let $q$ be a prime power and let $1\leq m\leq k\leq n$.  For a generator matrix
$G\in \F_q^{k\times n}$ of rank $k$, and for a subset $I\in \cI_m$, determine the expected value
\[
\E[\tau_I(G)].
\]
More generally, determine which generator matrices minimize this expectation, either for a fixed subset $I$, or uniformly over all subsets of cardinality $m$.  In particular, study the extremal quantities
\[
T_q^{\max}(n,k,m):=
\min_{\substack{G\in \F_q^{k\times n}\\ \rk(G)=k}}
\Tmax^{(m)}(G)
\]
and
\[
T_q^{\mathrm{avg}}(n,k,m):=
\min_{\substack{G\in \F_q^{k\times n}\\ \rk(G)=k}}
\Tavg^{(m)}(G).
\]
\end{problem}

The cases $m=1$ and $m=k$ correspond to two extremal regimes.  For $m=1$, the problem reduces to the classical random access problem for the recovery of a single information symbol.  For $m=k$, the problem becomes the recovery of the entire information vector and is closely related to the coverage depth problem.  Thus, the range
\[
1<m<k
\]
interpolates between local and global recovery.

The generalized random access problem asks how the expected number of random accesses changes as $m$ increases, and which geometric or coding-theoretic properties of $G$ control this transition.  In this sense, the intermediate values of $m$ provide a framework for understanding the structural features that are responsible for optimal or near-optimal behaviour in the two extremal cases.

In order to compute this expectation, it is useful to separate the probabilistic part of the sampling process from the linear-algebraic properties of the generator matrix.  This is done through the following counting function.

\begin{definition}
Let $0\leq s\leq n$.  We define
\[
\alpha_I(s):=
\left|
\left\{
S\subseteq [n] :
|S|=s,
\{e_i : i\in I\}\subseteq
\Span\{g_j : j\in S\}\rangle
\right\}
\right|.
\]
Equivalently, $\alpha_I(s)$ is the number of subsets of $s$ distinct columns of $G$ that are sufficient to recover all the information symbols indexed by $I$.  Since $I\neq\emptyset$, we set $\alpha_I(0)=0$.
\end{definition}

The function $\alpha_I(s)$ contains all the information about the geometry of the columns of $G$ that is relevant for the generalized random access problem.  Indeed, while the random sampling process is the same for every generator matrix with $n$ columns, the values of $\alpha_I(s)$ depend on how the columns of $G$ are arranged in $\F_q^k$.  The expectation of $\tau_I(G)$ can therefore be expressed purely in terms of these numbers.

This formulation is particularly useful because it treats the classical random access problem and the coverage depth problem within the same framework.  The cases $|I|=1$ and $|I|=k$ appear as two extremes of a single family of problems, while the intermediate values of $|I|$ measure the transition from local to global recovery.  In the next sections we exploit this point of view to compare different families of codes and to understand how their geometry affects the expected number of random accesses.

\subsection{Geometric formulation}\label{subsec:geometric-formulation}

The generalized random access problem can also be formulated in projective
geometric terms.  This point of view is useful because the recovery condition
depends only on linear spans of columns, and hence only on the corresponding
projective points.

Let $G\in \F_q^{k\times n}$ be a generator matrix of rank $k$, and assume
that its columns are nonzero.  This assumption is harmless for the geometric
formulation of the encoders considered below; moreover, in extremal questions
for fixed $n$, replacing a zero column by a nonzero column cannot increase any
of the expected recovery variables.  If
\[
g_1,\ldots,g_n\in \F_q^k
\]
are the columns of $G$, then each column defines a point
\[
P_j:=\langle g_j\rangle\in \PG(k-1,q).
\]
Thus $G$ determines a projective system
\[
\cG=\{P_1,\ldots,P_n\}\subseteq \PG(k-1,q),
\]
where repetitions are allowed if two columns are proportional.  Conversely,
choosing homogeneous representatives for the points of a rank-$k$ projective
system gives a generator matrix.  Since the recovery condition is invariant
under multiplication of columns by nonzero scalars, the random access variables
depend only on this projective system.

We denote by
\[
E_i:=\langle e_i\rangle,\qquad i\in[k],
\]
the fundamental points of $\PG(k-1,q)$.  For $I\subseteq[k]$, we write
\[
\cF_I:=\{E_i:i\in I\}.
\]
In the geometric language, recovering the information symbols indexed by
$I$ means that the span of the sampled projective points contains the set
$\cF_I$.

As before, the points of $\cG$ are sampled independently, uniformly at
random, and with replacement.  If $X_1,X_2,\ldots$ are independent random
variables uniformly distributed on $[n]$, we set
\[
S_r:=\{X_1,\ldots,X_r\}\subseteq[n].
\]
The corresponding geometric random access variable is
\[
\tau_I(\cG):=
\min\left\{
r\geq 1:
\cF_I\subseteq
\left\langle P_j:j\in S_r\right\rangle
\right\}.
\]
Here $\langle P_j:j\in S_r\rangle$ denotes the projective subspace generated
by the observed points.

Therefore the generalized random access problem is equivalent to the following
finite-geometric problem: given a rank-$k$ multiset
\[
\cG\subseteq \PG(k-1,q)
\]
containing the fundamental points, determine the expected number of random
samples needed until the sampled points generate a subspace containing a
prescribed set $\cF_I$ of fundamental points.

Equivalently, for fixed $1\leq m\leq k$, one may study
\[
\Tmax^{(m)}(\cG):=
\max_{I\in\cI_m}\E[\tau_I(\cG)]
\]
and
\[
\Tavg^{(m)}(\cG):=
\frac{1}{\binom{k}{m}}
\sum_{I\in\cI_m}\E[\tau_I(\cG)].
\]
The corresponding extremal problem asks for rank-$k$ projective systems of
size $n$ for which these quantities are as small as possible.

In this notation, the recovering-set counting function becomes
\[
\alpha_I(\cG,s):=
\left|
\left\{
S\subseteq[n]:
|S|=s,\ 
\cF_I\subseteq
\left\langle P_j:j\in S\right\rangle
\right\}
\right|.
\]
Thus $\alpha_I(\cG,s)$ counts the $s$-subsets of the projective system
whose span contains all requested fundamental points.  When the projective
system $\cG$ is clear from the context, we simply write $\alpha_I(s)$.

The two extremal cases have a simple geometric interpretation.  If
$|I|=1$, one asks when the sampled points generate a subspace containing a
fixed fundamental point.  If $I=[k]$, one asks when the sampled points
generate the whole space $\PG(k-1,q)$.  The intermediate case
$1<|I|<k$ asks when the sampled points generate a subspace containing a
prescribed coordinate subspace.  This is the geometric form of the
interpolation between local and global recovery.

\section{General formula for the expectation}\label{sec:general-expectation}

\subsection{The expectation formula}\label{subsec:expectation-formula}

The previous section separates the sampling process from the geometry of the columns. We now make this separation explicit. Throughout this section, for every integer $n\geq 1$, we denote by
\[
H_n:=\sum_{j=1}^n \frac{1}{j}
\]
the $n$-th harmonic number, and we set $H_0:=0$.

{The following subset-counting formula is implicit in the framework of \cite{gruica2024combinatorial} and is used explicitly in \cite{BarLev2026CodedRetrieval}.  We include a proof for completeness.  Once the numbers $\alpha_I(s)$ are known, the formula determines the expected value of the generalized random access variable for every code and applies uniformly to the two extremal cases $|I|=1$ and $I=[k]$, as well as to all intermediate values of $|I|$.}

\begin{lemma}\label{lem:general-expectation}
Let $G\in \F_q^{k\times n}$ be a generator matrix of rank $k$, and let $\emptyset\neq I\subseteq [k]$. Then
\begin{equation}\label{eq:general-expectation}
\E[\tau_I(G)]
= nH_n-\sum_{s=1}^{n-1}\frac{\alpha_I(s)}{\binom{n-1}{s}}.
\end{equation}
\end{lemma}

\begin{proof}
Since $G$ has rank $k$, the random variable $\tau_I(G)$ is almost surely finite. For a positive integer-valued random variable, we have
\[
\E[\tau_I(G)]
=
\sum_{r\geq 1}\Pr[\tau_I(G)\geq r].
\]
Let
\[
\eta_{r-1}:=|S_{r-1}|
\]
be the number of distinct columns observed after $r-1$ samples. Conditional on $\eta_{r-1}=s$, the set of observed columns is uniformly distributed among the $\binom ns$ subsets of $[n]$ of size $s$. Hence
\[
\Pr[\tau_I(G)\geq r\mid \eta_{r-1}=s]
=
1-\frac{\alpha_I(s)}{\binom ns},
\]
for $0\leq s\leq n-1$. Therefore,
\[
\E[\tau_I(G)]
=
\sum_{r\geq 1}
\sum_{s=0}^{n-1}
\left(1-\frac{\alpha_I(s)}{\binom ns}\right)
\Pr[\eta_{r-1}=s].
\]

The standard occupancy formula gives
\[
\Pr[\eta_{r-1}=s]
=
\binom ns
\sum_{j=0}^s
(-1)^j
\binom sj
\left(\frac{s-j}{n}\right)^{r-1}.
\]
Interchanging the finite sum in $s$ with the geometric series in $r$, we obtain
\[
\E[\tau_I(G)]
=
\sum_{s=0}^{n-1}
\bigl(\binom ns-\alpha_I(s)\bigr)
\sum_{j=0}^s
(-1)^j
\binom sj
\frac{n}{n-s+j}.
\]
Using the identity
\[
\sum_{j=0}^s
(-1)^j
\binom sj
\frac{n}{n-s+j}
=
\frac{1}{\binom{n-1}{s}},
\qquad 0\leq s\leq n-1,
\]
we get
\[
\E[\tau_I(G)]
=
\sum_{s=0}^{n-1}
\frac{\binom ns}{\binom{n-1}{s}}
-
\sum_{s=0}^{n-1}
\frac{\alpha_I(s)}{\binom{n-1}{s}}.
\]
Since
\[
\frac{\binom ns}{\binom{n-1}{s}}=\frac{n}{n-s},
\]
we have
\[
\sum_{s=0}^{n-1}
\frac{\binom ns}{\binom{n-1}{s}}
=
\sum_{s=0}^{n-1}\frac{n}{n-s}
=
nH_n.
\]
Finally, since $\alpha_I(0)=0$, formula \eqref{eq:general-expectation} follows.
\end{proof}

\subsection{Basic properties}\label{subsec:basic-properties}

We record some elementary properties of the numbers $\alpha_I(s)$.  They
will be useful later, and they also make explicit the interpolation between
local and global recovery.

\begin{proposition}\label{prop:alpha-properties}
Let $G\in \F_q^{k\times n}$ be a generator matrix of rank $k$, and let
$\emptyset\neq I\subseteq [k]$.  Set $m:=|I|$.  Then the following hold.
\begin{enumerate}[label=\textnormal{(\roman*)}]
\item For every $0\leq s\leq n$,
\[
0\leq \alpha_I(s)\leq \binom ns.
\]

\item If $s<m$, then
\[
\alpha_I(s)=0.
\]

\item We have
\[
\alpha_I(n)=1.
\]

\item If $I\subseteq J\subseteq [k]$, then, for every $s$,
\[
\alpha_J(s)\leq \alpha_I(s).
\]
Consequently,
\[
\tau_I(G)\leq \tau_J(G)
\]
pointwise, and therefore
\[
\E[\tau_I(G)]\leq \E[\tau_J(G)].
\]

\item The family of $I$-recovering subsets is upward closed: if
$S\subseteq [n]$ is $I$-recovering and $S\subseteq T\subseteq [n]$, then
$T$ is also $I$-recovering.  In particular, the normalized quantities
\[
p_I(s):=\frac{\alpha_I(s)}{\binom ns}
\]
are nondecreasing in $s$.
\end{enumerate}
\end{proposition}

\begin{proof}
The first assertion is immediate from the definition of $\alpha_I(s)$.  For
(ii), observe that the vectors $\{e_i:i\in I\}$ are linearly independent and
span an $m$-dimensional subspace of $\F_q^k$.  A set of $s<m$ columns
spans a space of dimension at most $s$, and therefore cannot contain all
these vectors.

Since $G$ has rank $k$, the full set of columns spans $\F_q^k$.  Hence
the unique subset of $[n]$ of size $n$ recovers every canonical vector,
proving (iii).

If $I\subseteq J$, then any set of columns which recovers all vectors
$e_j$, with $j\in J$, also recovers all vectors $e_i$, with $i\in I$.
Thus $\alpha_J(s)\leq \alpha_I(s)$ for every $s$.  The pointwise inequality
$\tau_I(G)\leq \tau_J(G)$ follows from the same inclusion of recovery
conditions, and taking expectations gives the desired inequality.

Finally, adding columns cannot destroy the recovery property, so the family
of $I$-recovering subsets is upward closed.  To prove that $p_I(s)$ is
nondecreasing, count pairs $(S,T)$, where $S$ is an $I$-recovering set
of size $s$, $T$ has size $s+1$, and $S\subseteq T$.  Each recovering
set $S$ of size $s$ is contained in $n-s$ subsets $T$ of size $s+1$,
and every such $T$ is again recovering.  On the other hand, each recovering
set $T$ of size $s+1$ contains at most $s+1$ subsets of size $s$.
Therefore
\[
\alpha_I(s)(n-s)\leq \alpha_I(s+1)(s+1),
\]
which is equivalent to
\[
\frac{\alpha_I(s)}{\binom ns}
\leq
\frac{\alpha_I(s+1)}{\binom n{s+1}}.
\]
This proves the claim.
\end{proof}

\begin{remark}\label{rem:interpolation-monotonicity}
The monotonicity in Proposition~\ref{prop:alpha-properties} reflects the role
of the generalized random access problem as an interpolation between local and
global recovery.  If $i\in I\subseteq [k]$, then
\[
\tau_i(G)\leq \tau_I(G)\leq \tau_{[k]}(G),
\]
and hence
\[
\E[\tau_i(G)]\leq \E[\tau_I(G)]\leq \E[\tau_{[k]}(G)].
\]
Thus, as the number of requested information symbols increases, the generalized
random access variable moves from the classical singleton problem toward the
full-recovery problem.
\end{remark}

\subsection{A dual-code viewpoint}\label{subsec:dual-viewpoint}

The counting function $\alpha_I(s)$ is defined in terms of the span of
subsets of columns of $G$.  Equivalently, for systematic encoders, it can be
interpreted in terms of linear dependencies among the columns, and hence in
terms of the dual code.

Throughout this subsection, assume that
\[
G=(I_k\mid A)
\]
is systematic, and let $\cC$ be the linear code generated by $G$.  Thus
\[
\cC^\perp=\{h\in \F_q^n : Gh^T=0\}
\]
is the space of linear dependencies among the columns of $G$.  We write
$\supp(h)$ for the support of a word $h\in \F_q^n$.

\begin{lemma}\label{lem:dual-recovery-single}
Let $G=(I_k\mid A)$ and let $i\in [k]$.  For a subset
$S\subseteq [n]$, one has
\[
e_i\in \Span\{g_j:j\in S\}\rangle
\]
if and only if either $i\in S$, or there exists a dual codeword
$h\in \cC^\perp$ such that
\[
h_i\neq 0
\qquad\text{and}\qquad
\supp(h)\subseteq S\cup\{i\}.
\]
\end{lemma}

\begin{proof}
If $i\in S$, then $e_i=g_i$ belongs to the span of the columns indexed by
$S$.  Suppose therefore that $i\notin S$.

Assume first that $e_i\in \Span\{g_j:j\in S\}\rangle$.  Then there exist
coefficients $\lambda_j\in \F_q$, for $j\in S$, such that
\[
e_i=\sum_{j\in S}\lambda_j g_j.
\]
Since $g_i=e_i$, this gives the linear dependency
\[
g_i-\sum_{j\in S}\lambda_j g_j=0.
\]
Hence the vector $h\in \F_q^n$ defined by
\[
h_i=1,
\qquad
h_j=-\lambda_j \ \text{for } j\in S,
\qquad
h_j=0 \ \text{otherwise},
\]
belongs to $\cC^\perp$, satisfies $h_i\neq 0$, and has support contained
in $S\cup\{i\}$.

Conversely, assume that there exists $h\in \cC^\perp$ with $h_i\neq 0$
and $\supp(h)\subseteq S\cup\{i\}$.  Since $Gh^T=0$, we have
\[
h_i g_i+\sum_{j\in S} h_j g_j=0.
\]
Using $g_i=e_i$ and $h_i\neq 0$, we obtain
\[
e_i=-\sum_{j\in S}\frac{h_j}{h_i}g_j.
\]
Therefore $e_i\in \Span\{g_j:j\in S\}\rangle$, as claimed.
\end{proof}

\begin{corollary}\label{cor:dual-recovery-I}
Let $G=(I_k\mid A)$, let $\emptyset\neq I\subseteq [k]$, and let
$S\subseteq [n]$.  Then $S$ is $I$-recovering if and only if, for every
$i\in I$, either $i\in S$, or there exists $h^{(i)}\in \cC^\perp$ such
that
\[
h^{(i)}_i\neq 0
\qquad\text{and}\qquad
\supp(h^{(i)})\subseteq S\cup\{i\}.
\]
Equivalently,
\[
\alpha_I(s)=
\left|
\left\{
S\subseteq [n] : |S|=s,
\forall i\in I,
\left(
 i\in S
 \text{ or }
 \exists h^{(i)}\in \cC^\perp
 \text{ with }
 h^{(i)}_i\neq 0,
 \supp(h^{(i)})\subseteq S\cup\{i\}
\right)
\right\}
\right|.
\]
\end{corollary}

\begin{remark}\label{rem:dual-low-weight}
The dual-code formulation shows that small values of $\E[\tau_I(G)]$ are
closely related to the presence of many low-weight dual codewords involving
the systematic positions indexed by $I$.  Indeed, if $i\notin S$, then the
only way to recover $e_i$ from the columns indexed by $S$ is through a
dual dependency whose support is contained in $S\cup\{i\}$.  Thus, dual
codewords provide alternative recovering sets besides the trivial ones obtained
by directly sampling the systematic coordinates.
\end{remark}

The minimum distance of the dual code gives a first obstruction to the
existence of small nontrivial recovering sets.

\begin{proposition}\label{prop:dual-distance-alpha}
Let $G=(I_k\mid A)$, let $\cC$ be the code generated by $G$, and let
$d^\perp=d(\cC^\perp)$.  Let $I\subseteq [k]$ with $|I|=m$.  Then
\[
\alpha_I(s)=0
\qquad\text{for }0\leq s<m.
\]
Moreover, for every $s$ such that
\[
m\leq s<d^\perp-1,
\]
we have
\[
\alpha_I(s)=\binom{n-m}{s-m}.
\]
\end{proposition}

\begin{proof}
The equality $\alpha_I(s)=0$ for $s<m$ follows from the fact that the
vectors $\{e_i:i\in I\}$ are linearly independent.

Now assume that $m\leq s<d^\perp-1$, and let $S\subseteq [n]$ be an
$I$-recovering set of size $s$.  We claim that $I\subseteq S$.  Suppose,
by contradiction, that there exists $i\in I\setminus S$.  Since $S$
recovers $e_i$, Lemma~\ref{lem:dual-recovery-single} gives a dual codeword
$h\in \cC^\perp$ such that
\[
h_i\neq 0
\qquad\text{and}\qquad
\supp(h)\subseteq S\cup\{i\}.
\]
Hence
\[
\operatorname{wt}(h)\leq |S|+1=s+1<d^\perp,
\]
which contradicts the definition of $d^\perp$.  Therefore $I\subseteq S$.

Conversely, if $I\subseteq S$, then all vectors $e_i$, with $i\in I$,
are directly available among the sampled columns, since $G$ is systematic.
Thus $S$ is $I$-recovering.  Hence, for $m\leq s<d^\perp-1$, the
$I$-recovering subsets of size $s$ are precisely the subsets $S$ of size
$s$ containing $I$.  Their number is
\[
\binom{n-m}{s-m}.
\]
This proves the claim.
\end{proof}

\begin{remark}\label{rem:dual-mds}
Proposition~\ref{prop:dual-distance-alpha} explains the role of the dual
minimum distance in the generalized random access problem.  If $d^\perp$ is
large, then there are no small nontrivial dual dependencies, and for small
values of $s$ the only way to recover the symbols indexed by $I$ is to
sample the corresponding systematic positions directly.  Improvements over
this trivial recovery mechanism can only appear starting from subsets of size
at least $d^\perp-1$.

For systematic MDS encoders, the dual code is again MDS and $d^\perp=k+1$.
Therefore Proposition~\ref{prop:dual-distance-alpha} gives
\[
\alpha_I(s)=\binom{n-m}{s-m}
\qquad
\text{for }m\leq s<k,
\]
which is exactly the first part of the counting used later in the MDS section.
In this sense, the MDS computation can be viewed as the case where no
nontrivial low-weight dual dependencies are available before dimension $k$.
\end{remark}

\subsection{Automorphisms and uniformity of the random access variables}
\label{subsec:automorphisms}

We now record a simple but useful consequence of symmetry. Since the
generalized random access variable depends on the chosen generator matrix, the
relevant automorphisms are those that preserve the projective system defined by
the columns of $G$ and act compatibly on the distinguished information
positions.

Let $G\in \F_q^{k\times n}$ have columns $g_1,\ldots,g_n$. We define the
information-preserving automorphism group of $G$, denoted by
$\Aut_{\mathrm{info}}(G)$, as the set of pairs $(A,\sigma)$, where
$A\in \mathrm{GL}_k(q)$ and $\sigma\in S_n$, such that for every
$j\in[n]$ there exists $\lambda_j\in \F_q^\ast$ with
\[
A g_j=\lambda_j g_{\sigma(j)},
\]
and such that $A$ permutes the one-dimensional subspaces generated by the
canonical basis vectors. Equivalently, there exists a permutation
$\pi\in S_k$ such that
\[
A\langle e_i\rangle=\langle e_{\pi(i)}\rangle
\qquad\text{for every } i\in[k].
\]
We denote by $\Pi(G)\leq S_k$ the permutation group induced on the
information positions.

The following observation shows that the expected value of $\tau_I(G)$ is
constant on the orbits of this group.

\begin{proposition}\label{prop:automorphism-orbits}
Let $G\in \F_q^{k\times n}$, and let $I,J\subseteq[k]$. Suppose that
there exists an element of $\Aut_{\mathrm{info}}(G)$ inducing a permutation
$\pi\in \Pi(G)$ such that
\[
J=\pi(I).
\]
Then, for every $0\leq s\leq n$,
\[
\alpha_I(s)=\alpha_J(s).
\]
Consequently,
\[
\E[\tau_I(G)]=\E[\tau_J(G)].
\]
\end{proposition}

\begin{proof}
Let $(A,\sigma)\in\Aut_{\mathrm{info}}(G)$ induce the permutation
$\pi\in S_k$, and assume that $J=\pi(I)$. For a subset
$S\subseteq[n]$, we have 
\[
\{e_i:i\in I\}\subseteq \langle\{g_\ell:\ell\in S\}\rangle
\]
if and only if
\[
A\langle e_i:i\in I\rangle
\subseteq
A\langle \{g_\ell:\ell\in S\}\rangle.
\]
Since $A$ maps the projective points $\langle e_i\rangle$, $i\in I$,
onto the projective points $\langle e_j\rangle$, $j\in J$, and maps the
columns indexed by $S$ onto scalar multiples of the columns indexed by
$\sigma(S)$, this is equivalent to
\[
\{e_j:j\in J\}\subseteq \langle\{g_\ell:\ell\in \sigma(S)\}\rangle.
\]
Thus $S$ is $I$-recovering if and only if $\sigma(S)$ is
$J$-recovering. Since $\sigma$ preserves cardinalities, it gives a
bijection between $I$-recovering subsets of size $s$ and $J$-recovering
subsets of size $s$. Hence $\alpha_I(s)=\alpha_J(s)$ for all $s$, and
the equality of expectations follows from Lemma~\ref{lem:general-expectation}.
\end{proof}

As a consequence, the computation of $\Tmax^{(m)}(G)$ and
$\Tavg^{(m)}(G)$ can be reduced to orbit representatives.

\begin{corollary}\label{cor:orbit-reduction}
Let
\[
\cI_m=\{I\subseteq[k]: |I|=m\},
\]
and let
\[
\cI_m=\mathcal O_1\cup\cdots\cup\mathcal O_r
\]
be the decomposition into orbits under the action of $\Pi(G)$. Choose one
representative $I_a\in\mathcal O_a$ for each $a=1,\ldots,r$. Then
\[
\Tmax^{(m)}(G)
=
\max_{1\leq a\leq r}
\E[\tau_{I_a}(G)]
\]
and
\[
\Tavg^{(m)}(G)
=
\frac{1}{\binom{k}{m}}
\sum_{a=1}^r
|\mathcal O_a|\,\E[\tau_{I_a}(G)].
\]
\end{corollary}

In particular, if $\Pi(G)$ is transitive on the $m$-subsets of $[k]$,
then the generalized random access variable depends only on $m$, and not on
the particular choice of $I$.

\begin{corollary}\label{cor:m-homogeneous}
Suppose that the induced group $\Pi(G)$ is $m$-homogeneous, i.e. it acts
transitively on the subsets of $[k]$ of cardinality $m$. Then
\[
\E[\tau_I(G)]
\]
is the same for every $I\subseteq[k]$ with $|I|=m$. In particular,
\[
\Tmax^{(m)}(G)=\Tavg^{(m)}(G)=\E[\tau_I(G)]
\]
for any $I\in\cI_m$.
\end{corollary}

\begin{remark}
This is one of the reasons why highly symmetric encoders are particularly
tractable. If the induced automorphism group is $m$-homogeneous for every
$m$, then the whole generalized random access problem is governed by a
single sequence
\[
E_m(G):=\E[\tau_I(G)],\qquad |I|=m,
\]
with
\[
E_1(G)\leq E_2(G)\leq \cdots \leq E_k(G).
\]
Thus the interpolation between singleton random access and full recovery can
be studied without distinguishing among different subsets $I$ of the same
cardinality.
\end{remark}

\subsection{Two elementary examples}\label{subsec:elementary-examples}

We include two elementary examples which clarify the role of the notation
introduced above.  The first one shows that the generalized random access
variable may depend on the chosen set $I$, even when $|I|$ is fixed.  The
second one gives a closed formula for a simple family of systematic codes.

\begin{example}\label{ex:asymmetric}
Let
\[
G=\bigl(e_1,e_2,\ldots,e_k,e_1\bigr)\in \F_q^{k\times (k+1)}.
\]
Thus the first information coordinate appears twice among the columns, while
each of the other information coordinates appears only once.  For singleton
recovery, we have
\[
\E[\tau_{\{1\}}(G)]=\frac{k+1}{2},
\]
because one has to sample one of two columns equal to $e_1$.  On the other
hand, for every $i\in\{2,\ldots,k\}$,
\[
\E[\tau_{\{i\}}(G)]=k+1,
\]
because the only way to recover $e_i$ is to sample the unique column equal
to $e_i$.

Therefore the expected value of the random access variable is not determined
only by the cardinality of $I$.  In particular,
\[
T_{\max}^{(1)}(G)=k+1
\qquad\text{and}\qquad
T_{\mathrm{avg}}^{(1)}(G)
=
\frac{1}{k}
\left(
\frac{k+1}{2}+(k-1)(k+1)
\right).
\]
This example illustrates why it is natural to distinguish between the maximum
and the average quantities.
\end{example}

\begin{example}\label{ex:single-parity}
Let
\[
G=\bigl(I_k\mid \mathbf 1\bigr)\in \F_q^{k\times (k+1)},
\]
where $\mathbf 1=e_1+\cdots+e_k$.  We assume that
$\operatorname{char}(\F_q)$ is arbitrary; the following argument only uses
the fact that the last column is the sum of all systematic columns.

Let $I\subseteq[k]$ with $|I|=m$.  We compute the corresponding numbers
$\alpha_I(s)$.  If the parity column is not selected, then a set of columns
recovers the vectors $e_i$, $i\in I$, if and only if it contains all the
systematic columns indexed by $I$.  If the parity column is selected, then
the set also recovers all information symbols whenever all but at most one of
the systematic columns have been selected.

Hence, for $0\leq s\leq k+1$, one obtains
\[
\alpha_I(s)=
\begin{cases}
0, & 0\leq s<m,\\[2mm]
\binom{k-m+1}{s-m}, & m\leq s\leq k-1,\\[2mm]
k+1, & s=k,\\[2mm]
1, & s=k+1.
\end{cases}
\]
Therefore, by Lemma~\ref{lem:general-expectation},
\[
\E[\tau_I(G)]
=
(k+1)H_{k+1}
-
\sum_{s=m}^{k-1}
\frac{\binom{k-m+1}{s-m}}{\binom{k}{s}}
-
(k+1).
\]
Equivalently,
\[
\E[\tau_I(G)]
=
(k+1)(H_{k+1}-1)
-
\sum_{s=m}^{k-1}
\frac{\binom{k-m+1}{s-m}}{\binom{k}{s}}.
\]

For $m=1$, this gives
\[
\E[\tau_i(G)]=k
\]
for every $i\in[k]$, recovering the known singleton behaviour of the
systematic single-parity code.  For $m=k$, the expression gives
\[
\E[\tau_{[k]}(G)]
=
(k+1)(H_{k+1}-1),
\]
which is the expected number of samples needed to collect any $k$ out of the
$k+1$ columns.
\end{example}

\section{General bounds}\label{sec:bounds}

The expectation formula of Lemma~\ref{lem:general-expectation} is exact, but it requires detailed knowledge of all the numbers $\alpha_I(s)$.  In general, this is a difficult finite-geometric counting problem.  We therefore collect bounds which depend only on coarser information, such as the size of minimal recovering sets and the order in which individual coordinates become recoverable.  These bounds hold for arbitrary full-rank generator matrices and generalize the bounds for the singleton random-access problem, where one considers only the case $m=1$; see, for instance, the parameters $\Tmax$ and $\Tavg$ and the singleton lower bounds in \cite{bar2024cover}.

\begin{definition}\label{def:minimal-recovering-sets}
Let $G\in \F_q^{k\times n}$ have rank $k$, and let $\emptyset\neq I\subseteq [k]$.  A subset $A\subseteq [n]$ is called an $I$-recovering set if
\[
\{e_i:i\in I\}\subseteq \langle\{g_j:j\in A\}\rangle.
\]
It is called minimal if no proper subset of $A$ is $I$-recovering.  We denote by $\cD_I(G)$ the family of minimal $I$-recovering sets, and we set
\[
\rho_I(G):=\min\{|A|:A\in \cD_I(G)\}.
\]
For fixed $m$, we also write
\[
\rho_{\max}^{(m)}(G):=\max_{I\in \cI_m}\rho_I(G).
\]
\end{definition}

Since the vectors $\{e_i:i\in I\}$ are linearly independent, every $I$-recovering set has size at least $|I|$.  Since $G$ has rank $k$, there is always a set of at most $k$ columns spanning the whole ambient space.  Therefore,
\[
|I|\leq \rho_I(G)\leq k.
\]

\begin{proposition}\label{prop:pointwise-bounds}
Let $G\in \F_q^{k\times n}$ have rank $k$, and let $\emptyset\neq I\subseteq [k]$.  Then
\begin{equation}\label{eq:pointwise-rho-bound}
n\bigl(H_n-H_{n-\rho_I(G)}\bigr)
\leq
\E[\tau_I(G)]
\leq
nH_{\rho_I(G)}.
\end{equation}
In particular, if $|I|=m$, then
\begin{equation}\label{eq:pointwise-m-bound}
n\bigl(H_n-H_{n-m}\bigr)
\leq
\E[\tau_I(G)]
\leq
nH_k.
\end{equation}
Moreover,
\begin{equation}\label{eq:tmax-upper-rho}
\Tmax^{(m)}(G)\leq nH_{\rho_{\max}^{(m)}(G)}\leq nH_k.
\end{equation}
If $G$ is systematic, then $\rho_I(G)=|I|$ for every $I\subseteq [k]$, and hence
\begin{equation}\label{eq:systematic-upper}
\Tmax^{(m)}(G)\leq nH_m.
\end{equation}
\end{proposition}

\begin{proof}
Let $\rho:=\rho_I(G)$.  Before $\rho$ distinct columns have been observed, no $I$-recovering set can be contained in the observed set.  Hence $\tau_I(G)$ is at least the coupon-collector variable which counts the number of samples needed to observe $\rho$ distinct columns among $n$.  Its expectation is $n(H_n-H_{n-\rho})$, proving the lower bound in \eqref{eq:pointwise-rho-bound}.

For the upper bound, choose a minimal $I$-recovering set $A$ of size $\rho$.  Once all columns indexed by $A$ have been sampled, all symbols indexed by $I$ can be recovered.  The expected number of samples needed to collect this fixed set $A$ is $nH_\rho$.  Therefore $\E[\tau_I(G)]\leq nH_\rho$.

The inequalities in \eqref{eq:pointwise-m-bound} follow from $m\leq \rho_I(G)\leq k$.  Taking the maximum over all $I\in \cI_m$ gives \eqref{eq:tmax-upper-rho}.  Finally, if $G$ is systematic, the $m$ systematic columns indexed by $I$ recover precisely the $m$ requested independent vectors, so $\rho_I(G)=m$.
\end{proof}

The next result extends the lower bounds for the singleton parameter $\Tmax$ to all values of $m$.  The proof uses the order statistics of the one-coordinate variables.

\begin{theorem}\label{thm:universal-lower-bounds}
Let $G\in \F_q^{k\times n}$ have rank $k$, and let $1\leq m\leq k$.  Then
\begin{equation}\label{eq:rate-sensitive-lower}
\Tmax^{(m)}(G)
\geq
\Tavg^{(m)}(G)
\geq
L(n,k,m),
\end{equation}
where
\begin{equation}\label{eq:def-Lnkm}
L(n,k,m):=
\frac{n}{\binom{k}{m}}
\sum_{j=m}^{k}
\binom{j-1}{m-1}
\bigl(H_n-H_{n-j}\bigr).
\end{equation}
Consequently,
\begin{equation}\label{eq:combinatorial-lower}
\Tmax^{(m)}(G)
\geq
\Tavg^{(m)}(G)
\geq
\frac{m(k+1)}{m+1}.
\end{equation}
\end{theorem}

\begin{proof}
For a fixed infinite sampling sequence, write
\[
\tau_i:=\tau_{\{i\}}(G),\qquad i\in [k],
\]
and let
\[
\tau_{(1)}\leq \tau_{(2)}\leq \cdots\leq \tau_{(k)}
\]
be the corresponding order statistics.  Let $C_j$ be the number of samples needed to observe $j$ distinct columns.  Before the sample in which $C_j$ is reached, the observed columns span a space of dimension at most $j-1$, and therefore cannot contain $j$ linearly independent canonical vectors.  Hence
\[
\tau_{(j)}\geq C_j
\qquad\text{for every }j\in [k].
\]
Moreover,
\[
\E[C_j]=n(H_n-H_{n-j}).
\]

Since
\[
\tau_I(G)=\max_{i\in I}\tau_i(G),
\]
we have, for each fixed sampling sequence,
\[
\frac{1}{\binom{k}{m}}\sum_{I\in \cI_m}\tau_I(G)
=
\frac{1}{\binom{k}{m}}
\sum_{j=m}^{k}\binom{j-1}{m-1}\tau_{(j)}.
\]
Indeed, $\tau_{(j)}$ is the maximum of an $m$-subset exactly when the remaining $m-1$ elements are chosen among the previous $j-1$ order statistics.  Taking expectations and using $\tau_{(j)}\geq C_j$ gives \eqref{eq:rate-sensitive-lower}.

Finally, since $n(H_n-H_{n-j})\geq j$, we obtain
\[
L(n,k,m)
\geq
\frac{1}{\binom{k}{m}}
\sum_{j=m}^{k}\binom{j-1}{m-1}j.
\]
The last expression is the average maximum of an $m$-subset of $[k]$, which is
\[
\frac{m(k+1)}{m+1}.
\]
This proves \eqref{eq:combinatorial-lower}.
\end{proof}

\begin{remark}\label{rem:bounds-extremes}
For $m=1$, Theorem~\ref{thm:universal-lower-bounds} gives
\[
\Tmax^{(1)}(G)
\geq
\Tavg^{(1)}(G)
\geq
\frac{n}{k}\sum_{j=1}^{k}\bigl(H_n-H_{n-j}\bigr)
=
\frac{n}{k}\sum_{a=0}^{k-1}\frac{k-a}{n-a}.
\]
Equivalently,
\[
\Tmax^{(1)}(G)
\geq
n-\frac{n(n-k)}{k}\bigl(H_n-H_{n-k}\bigr),
\]
which is the rate-sensitive singleton lower bound of \cite{bar2024cover} in the present notation.  The weaker bound \eqref{eq:combinatorial-lower} becomes
\[
\Tmax^{(1)}(G)\geq \frac{k+1}{2},
\]
again recovering the corresponding singleton bound.

For $m=k$, Theorem~\ref{thm:universal-lower-bounds} gives
\[
\Tmax^{(k)}(G)=\Tavg^{(k)}(G)
\geq n(H_n-H_{n-k}),
\]
which is the usual lower bound for full recovery.  Thus the bounds interpolate between the singleton random-access regime and the full-recovery coverage-depth regime.
\end{remark}

Combining the lower bounds with the systematic upper bound gives the following general estimate for the extremal quantities.

\begin{corollary}\label{cor:extremal-general-bounds}
For every prime power $q$ and every $1\leq m\leq k\leq n$,
\[
L(n,k,m)
\leq
T_q^{\mathrm{avg}}(n,k,m)
\leq
T_q^{\max}(n,k,m)
\leq
nH_m.
\]
If a systematic MDS encoder with parameters $[n,k]$ over $\F_q$ exists, then the last upper bound can be improved to the MDS value given in Theorem~\ref{thm:mds-general}.
\end{corollary}

We also record the following exact formula, which generalizes the disjoint-minimal-retrieval-set computation for one information symbol.

\begin{proposition}\label{prop:disjoint-recovering-sets}
Let $G\in \F_q^{k\times n}$ have rank $k$, let $\emptyset\neq I\subseteq [k]$, and suppose that
\[
\cD_I(G)=\{A_1,\ldots,A_v\}
\]
with the sets $A_1,
\ldots,A_v$ mutually disjoint.  Then
\begin{equation}\label{eq:disjoint-recovering-exact}
\E[\tau_I(G)]
=
n\sum_{\emptyset\neq J\subseteq [v]}
(-1)^{|J|+1}
H_{\sum_{j\in J}|A_j|}.
\end{equation}
In particular, if $v=1$, then $\E[\tau_I(G)]=nH_{|A_1|}$.
\end{proposition}

\begin{proof}
Let $Y_A$ denote the number of samples needed to collect all columns in a fixed set $A\subseteq [n]$.  Since the sets $A_1,\ldots,A_v$ are disjoint, the event that all sets indexed by $J\subseteq [v]$ have been collected is the event that the fixed set $\bigcup_{j\in J}A_j$ has been collected.  The expectation of the corresponding coupon-collector variable is
\[
nH_{|\cup_{j\in J}A_j|}
=
nH_{\sum_{j\in J}|A_j|}.
\]
The variable $\tau_I(G)$ is the minimum of the variables $Y_{A_1},\ldots,Y_{A_v}$.  Applying inclusion--exclusion to the tail-sum formula for this minimum gives \eqref{eq:disjoint-recovering-exact}.
\end{proof}

\section{Systematic MDS encoders}\label{sec:mds}

In this section $G\in \F_q^{k\times n}$ is assumed to be a systematic generator matrix of an MDS code, so that, after a possible permutation of columns,
\[
G=(I_k\mid A).
\]
This hypothesis is natural for random access to information symbols: the vectors $e_1,\ldots,e_k$ appear among the columns and represent the uncoded information positions.

For $1\leq m\leq k$, define
\begin{equation}\label{eq:def-mds-value}
M_{n,k}(m):=
k+\sum_{s=1}^{k-1}\frac{s}{n-s}
\left(1-\frac{\binom{s-1}{m-1}}{\binom{n-1}{m-1}}\right),
\end{equation}
where, as usual, $\binom{a}{b}=0$ when $b>a$.

{The following expression is algebraically equivalent to the systematic-MDS formula in \cite[Propositions~23 and~24]{BarLev2026CodedRetrieval}.  We record it in a form suited to varying the cardinality $m$ and to the uniform parameters $\Tmax^{(m)}$ and $\Tavg^{(m)}$.}

\begin{theorem}\label{thm:mds-general}
Let $G\in \F_q^{k\times n}$ be a systematic generator matrix of an MDS code and let $I\in \cI_m$, with $1\leq m\leq k$.  Then
\begin{equation}\label{eq:mds-general}
\E[\tau_I(G)]
=
M_{n,k}(m).
\end{equation}
Consequently,
\[
\Tmax^{(m)}(G)=\Tavg^{(m)}(G)=M_{n,k}(m).
\]
\end{theorem}

\begin{proof}
For $s<k$, every set of $s$ columns is linearly independent.  Hence the $m$ vectors $e_i$, $i\in I$, are contained in the span of an $s$-subset of columns if and only if those $m$ systematic columns have been selected.  Therefore
\[
\alpha_I(s)=
\begin{cases}
0, & 1\le s<m,\\[2mm]
\binom{n-m}{s-m}, & m\le s<k,\\[2mm]
\binom ns, & k\le s\le n.
\end{cases}
\]
Substitution in Lemma~\ref{lem:general-expectation} gives
\[
\E[\tau_I(G)]
=nH_n-
\sum_{s=m}^{k-1}\frac{\binom{n-m}{s-m}}{\binom{n-1}{s}}
-
\sum_{s=k}^{n-1}\frac{\binom ns}{\binom{n-1}{s}}.
\]
The last sum is $nH_{n-k}$.  Moreover,
\[
n(H_n-H_{n-k})=k+\sum_{s=1}^{k-1}\frac{s}{n-s}.
\]
Finally,
\[
\frac{\binom{n-m}{s-m}}{\binom{n-1}{s}}
=\frac{s}{n-s}\frac{\binom{s-1}{m-1}}{\binom{n-1}{m-1}},
\]
where the binomial coefficient is interpreted as zero when $s<m$.  This yields \eqref{eq:mds-general}.  The expression depends only on $m$, and therefore the maximum and average over $\cI_m$ coincide.
\end{proof}

\begin{proposition}\label{prop:mds-monotone}
For systematic MDS encoders, the quantity $M_{n,k}(m)$ is strictly increasing as a function of $m$, for $1\leq m<k$.
\end{proposition}

\begin{proof}
Using
\[
\frac{\binom{s-1}{m}}{\binom{n-1}{m}}
=\frac{s-m}{n-m}\frac{\binom{s-1}{m-1}}{\binom{n-1}{m-1}},
\]
we obtain
\[
M_{n,k}(m+1)-M_{n,k}(m)
=\frac{1}{n-m}\sum_{s=1}^{k-1}s\frac{\binom{s-1}{m-1}}{\binom{n-1}{m-1}}>0.
\]
The strict positivity follows from the term $s=m$.
\end{proof}

\begin{corollary}\label{cor:mds-extremes}
Let $G$ be a systematic MDS generator matrix.  Then
\[
M_{n,k}(1)=k
\]
and
\[
M_{n,k}(k)=n(H_n-H_{n-k}).
\]
Consequently, for $1<m<k$,
\[
k<\E[\tau_I(G)]<n(H_n-H_{n-k}).
\]
Moreover, for fixed $k$ and fixed $m$,
\[
\lim_{n\to\infty}M_{n,k}(m)=k.
\]
\end{corollary}

\begin{proof}
For $m=1$, the binomial quotient in \eqref{eq:def-mds-value} is equal to $1$ for every $s$, so all summands vanish and $M_{n,k}(1)=k$.  For $m=k$, one has $\binom{s-1}{k-1}=0$ for $1\leq s\leq k-1$, and hence
\[
M_{n,k}(k)=k+\sum_{s=1}^{k-1}\frac{s}{n-s}=n(H_n-H_{n-k}).
\]
The strict inequalities for $1<m<k$ follow from Proposition~\ref{prop:mds-monotone}.  The limit follows from
\[
0\leq M_{n,k}(m)-k\leq \sum_{s=1}^{k-1}\frac{s}{n-s},
\]
whose right-hand side tends to zero as $n\to\infty$.
\end{proof}

Thus, in the MDS case, the generalized random access problem gives a monotone scale between the two extremal values.  The single-coordinate value is $k$, while the full-recovery value is $n(H_n-H_{n-k})$; the intermediate expectations quantify how quickly the model moves from the local regime to the global one.  Notice also that for $m=k$ the MDS value coincides with the universal lower bound in Theorem~\ref{thm:universal-lower-bounds}, recovering the optimality of systematic MDS encoders for full recovery.

\begin{remark}
For $k=3$ the formula becomes especially simple.  If $|I|=2$, then
\begin{equation}\label{eq:mds-k3-m2}
\E[\tau_I(G)]=3+\frac{3}{n-1},
\end{equation}
and if $|I|=3$, then
\begin{equation}\label{eq:mds-k3-m3}
\E[\tau_I(G)]=3+\frac{1}{n-1}+\frac{2}{n-2}.
\end{equation}
\end{remark}

\section{Simplex encoders in arbitrary dimension}\label{sec:simplex-general}

Let $G_{\mathrm{Simp}}$ be a generator matrix of the $q$-ary simplex code of dimension $k$, obtained by choosing one representative for each point of the projective space $\PG(k-1,q)$.  Thus the length is
\[
n=\frac{q^k-1}{q-1}.
\]
Equivalently, the projective system of $G_{\mathrm{Simp}}$ is the whole space $\PG(k-1,q)$.

For $0\leq r\leq k$, we set
\[
N_r:=|\PG(r-1,q)|=\frac{q^r-1}{q-1},
\]
with the convention $N_0=0$.  We denote by
\[
\qbinom{a}{b}
\]
the Gaussian binomial coefficient, i.e. the number of $b$-dimensional subspaces of an $a$-dimensional vector space over $\F_q$.

Let $I\subseteq[k]$ with $|I|=m$, and let
\[
U_I:=\Span\{e_i:i\in I\}\rangle
\]
be the corresponding $m$-dimensional coordinate subspace of $\F_q^k$.  Since the projective system of the simplex code is the whole projective space, the expectation depends only on $m=|I|$, and not on the particular subset $I$.  We denote this common value by
\[
E_m^{\mathrm{Simp}}(q,k):=\E[\tau_I(G_{\mathrm{Simp}})],
\qquad |I|=m.
\]

We first compute the numbers $\alpha_I(s)$.  For $0\leq r\leq k$, let $\beta_q(r,s)$ be the number of $s$-subsets of $\PG(r-1,q)$ that span the whole $(r-1)$-dimensional projective space.  Equivalently, $\beta_q(r,s)$ is the number of $s$-subsets of the points of an $r$-dimensional vector space whose linear span has dimension $r$.

\begin{lemma}\label{lem:beta-simplex}
For $0\leq r\leq k$ and $0\leq s\leq N_r$, one has
\[
\beta_q(r,s)
=
\sum_{t=0}^{r}
(-1)^{r-t}
q^{\binom{r-t}{2}}
\qbinom{r}{t}
\binom{N_t}{s}.
\]
\end{lemma}

\begin{proof}
Fix an $r$-dimensional vector space $W$.  For each subspace $T\leq W$, the number of $s$-subsets of projective points contained in $T$ is $\binom{N_{\dim T}}{s}$.  By M\"obius inversion in the lattice of subspaces of $W$, the number of $s$-subsets whose span is exactly $W$ is
\[
\sum_{T\leq W}
\mu(T,W)\binom{N_{\dim T}}{s},
\]
where
\[
\mu(T,W)=
(-1)^{r-\dim T}q^{\binom{r-\dim T}{2}}.
\]
Since the number of $t$-dimensional subspaces of $W$ is $\qbinom{r}{t}$, the claimed formula follows.
\end{proof}

\begin{proposition}\label{prop:alpha-simplex-general}
Let $G_{\mathrm{Simp}}$ be the $q$-ary simplex encoder of dimension $k$, and let $I\subseteq[k]$ with $|I|=m$.  Then, for every $0\leq s\leq n$,
\[
\alpha_I(s)
=
\sum_{r=m}^{k}
\qbinom{k-m}{r-m}
\beta_q(r,s),
\]
where $\beta_q(r,s)$ is given in Lemma~\ref{lem:beta-simplex}.
\end{proposition}

\begin{proof}
Let $S$ be an $s$-subset of the projective system of the simplex code, and let $W=\Span\{P_j:j\in S\}\rangle$.  The set $S$ recovers all symbols indexed by $I$ if and only if
\[
U_I\leq W.
\]
Suppose that $\dim W=r$.  The number of $r$-dimensional subspaces $W\leq \F_q^k$ containing $U_I$ is
\[
\qbinom{k-m}{r-m}.
\]
For each such $W$, the number of $s$-subsets of the projective points of $W$ spanning $W$ is $\beta_q(r,s)$.  Summing over all possible dimensions $r=m,\ldots,k$, we obtain the formula.
\end{proof}

Combining Proposition~\ref{prop:alpha-simplex-general} with the general expectation formula gives the following closed expression.

\begin{theorem}\label{thm:simplex-general-expectation}
Let $G_{\mathrm{Simp}}$ be the $q$-ary simplex encoder of dimension $k$, and let $I\subseteq[k]$ with $|I|=m$.  Then
\[
E_m^{\mathrm{Simp}}(q,k)
=
nH_n
-
\sum_{s=1}^{n-1}
\frac{1}{\binom{n-1}{s}}
\sum_{r=m}^{k}
\qbinom{k-m}{r-m}
\sum_{t=0}^{r}
(-1)^{r-t}
q^{\binom{r-t}{2}}
\qbinom{r}{t}
\binom{N_t}{s}.
\]
In particular,
\[
\Tmax^{(m)}(G_{\mathrm{Simp}})
=
\Tavg^{(m)}(G_{\mathrm{Simp}})
=
E_m^{\mathrm{Simp}}(q,k).
\]
\end{theorem}

\begin{proof}
The formula follows by substituting the expression for $\alpha_I(s)$ from Proposition~\ref{prop:alpha-simplex-general} into Lemma~\ref{lem:general-expectation}.  The equality between maximum and average follows from symmetry.  Indeed, the monomial projective transformations induced by permutations of the coordinate axes preserve the simplex projective system and act transitively on the subsets of fundamental points of any fixed cardinality.
\end{proof}

\begin{remark}\label{rem:simplex-extremes}
The formula in Theorem~\ref{thm:simplex-general-expectation} interpolates between the singleton random-access problem and the full-recovery coverage-depth problem.

For $m=1$, the simplex code is recovery balanced, and one obtains
\[
E_1^{\mathrm{Simp}}(q,k)=k.
\]
Thus the simplex encoder behaves, for singleton recovery, like the systematic MDS encoder and the identity encoder.

For $m=k$, the condition is that the observed columns span the whole space $\F_q^k$.  In this case
\[
E_k^{\mathrm{Simp}}(q,k)
=
k+
\sum_{i=1}^{k}
\frac{q^{i-1}-1}{q^k-q^{i-1}}.
\]
This is the full-recovery coverage-depth value of the $q$-ary simplex code, as computed in \cite{bertuzzo2026duality}.  Equivalently,
\[
E_k^{\mathrm{Simp}}(q,k)
=
\sum_{i=0}^{k-1}
\frac{N_k}{N_k-N_i}.
\]
Indeed, when the currently observed span has dimension $i$, the probability that the next sampled point increases the dimension is
\[
\frac{N_k-N_i}{N_k}.
\]
\end{remark}

For actual computations, the following recursive form is often more convenient than the closed formula above.  It extends the dimension-three incidence argument by keeping track not only of the dimension of the observed span, but also of its intersection with the requested coordinate subspace.

\begin{proposition}\label{prop:simplex-recursion}
Let $U_I=\Span\{e_i:i\in I\}\rangle$, with $\dim U_I=m$.  For integers $a,b$ with $0\leq a\leq m$ and $0\leq b\leq k$, let $R_{a,b}$ denote the expected number of further samples needed to obtain a span containing $U_I$, starting from a subspace $L\leq \F_q^k$ such that
\[
\dim L=b,
\qquad
\dim(L\cap U_I)=a.
\]
Then
\[
R_{m,b}=0
\]
for every $b\geq m$.  For $a<m$ and for every admissible state with $b<k$, one has
\[
R_{a,b}
=
\frac{N_k}{N_k-N_b}
+
\frac{N_{b+m-a}-N_b}{N_k-N_b}R_{a+1,b+1}
+
\frac{N_k-N_{b+m-a}}{N_k-N_b}R_{a,b+1}.
\]
The desired expectation is
\[
E_m^{\mathrm{Simp}}(q,k)=R_{0,0}.
\]
In the recursion, terms with zero coefficient are omitted; equivalently, no
state with $b=k$ and $a<m$ is ever evaluated.
\end{proposition}

\begin{proof}
Assume that the current observed span is $L$, with
\[
\dim L=b,
\qquad
\dim(L\cap U_I)=a<m.
\]
There are $N_b$ projective points already contained in $L$; drawing one of them does not change the state.  Thus, after conditioning on drawing a point outside $L$, the expected waiting contribution is
\[
\frac{N_k}{N_k-N_b}.
\]

Now consider a point $P\notin L$.  The new span is $L+\langle P\rangle$, which has dimension $b+1$.  The dimension of the intersection with $U_I$ increases from $a$ to $a+1$ precisely when
\[
P\in L+U_I
\]
but $P\notin L$.  Since
\[
\dim(L+U_I)=b+m-a,
\]
the number of such projective points is
\[
N_{b+m-a}-N_b.
\]
The remaining points outside $L+U_I$, namely
\[
N_k-N_{b+m-a},
\]
increase the dimension of $L$ but do not increase the intersection with $U_I$.  This gives the stated recursion.  Starting from the zero subspace corresponds to the state $(a,b)=(0,0)$, hence the desired expectation is $R_{0,0}$.
\end{proof}

\begin{example}\label{ex:simplex-small-values}
For the binary simplex encoder of dimension $3$, one obtains
\[
E_1^{\mathrm{Simp}}(2,3)=3,
\qquad
E_2^{\mathrm{Simp}}(2,3)=\frac{11}{3},
\qquad
E_3^{\mathrm{Simp}}(2,3)=\frac{47}{12}.
\]
For the ternary simplex encoder of dimension $3$, one obtains
\[
E_1^{\mathrm{Simp}}(3,3)=3,
\qquad
E_2^{\mathrm{Simp}}(3,3)=\frac{41}{12},
\qquad
E_3^{\mathrm{Simp}}(3,3)=\frac{127}{36}.
\]
These values agree with the dimension-three computation obtained by counting lines in $\PG(2,q)$, while the formulae above work uniformly for all dimensions $k$.
\end{example}

\section{Balanced quasi-arcs in dimension three}\label{sec:dimension-three}

We now specialize to $k=3$ and use the language of the projective plane $\PG(2,q)$.  This is the first case in which the generalized problem has a genuinely intermediate value of $|I|$, namely $|I|=2$.  It is therefore the simplest setting where one can see how the model connects the single-symbol and full-recovery extremal cases.  The requested information vectors are represented by three non-collinear fundamental points
\[
E_1=(1:0:0),\qquad E_2=(0:1:0),\qquad E_3=(0:0:1).
\]

Balanced quasi-arcs were introduced in \cite{gruica2024geometry} to create projective systems in which the fundamental points are intentionally easier to recover.

\begin{definition}
Let $\cF=\{E_1,E_2,E_3\}\subseteq \PG(2,q)$ be three non-collinear points.  A set $\cG\subseteq \PG(2,q)$ is a balanced quasi-arc of weight $x$ if
\[
\cG=\cF\cup \cG_1\cup \cG_2\cup \cG_3,
\]
where $|\cG_1|=|\cG_2|=|\cG_3|=x$ and, up to relabelling,
\[
\cG_1\subseteq E_1E_2,\qquad
\cG_2\subseteq E_2E_3,\qquad
\cG_3\subseteq E_3E_1.
\]
Moreover, every line distinct from the three fundamental lines meets $\cG$ in at most two points.
\end{definition}

Thus $|\cG|=3x+3$.  The three fundamental lines are $(x+2)$-secants, whereas all other lines are at most $2$-secants.  Any balanced quasi-arc gives a $3\times (3x+3)$ generator matrix by taking homogeneous representatives of the points of $\cG$ as columns.

\begin{proposition}\label{prop:quasi-alpha-single}
Let $\cG_x$ be a balanced quasi-arc of weight $x$ and let $I=\{i\}$.  Then the value of $\alpha_I(s)$ is independent of the chosen fundamental point and is given by
\[
\alpha_I(s)=
\begin{cases}
1, & s=1,\\[1mm]
2\binom{x+2}{2}+x, & s=2,\\[1mm]
\binom{3x+3}{s}-\binom{x+2}{s}, & 3\le s\le x+2,\\[1mm]
\binom{3x+3}{s}, & x+2<s\le 3x+2.
\end{cases}
\]
Consequently,
\begin{equation}\label{eq:quasi-m1}
\E[\tau_I(\cG_x)]
=3+\frac{2}{3x+1}
-\frac{2(x+2)(x+1)+2x}{(3x+2)(3x+1)}
+\sum_{s=3}^{x+2}\prod_{j=0}^{s-1}\frac{x+2-j}{3x+2-j}.
\end{equation}
\end{proposition}

\begin{proof}
For $s=1$, only the fundamental point itself works.  For $s=2$, one either selects two points on one of the two fundamental lines through the chosen point, or one selects the chosen fundamental point together with a point on the opposite fundamental line; this gives $2\binom{x+2}{2}+x$.  For $s\ge 3$, the only non-recovering subsets are those contained in the fundamental line opposite to the chosen point.  Substituting these values in Lemma~\ref{lem:general-expectation} and using $n=3x+3$ gives \eqref{eq:quasi-m1}.
\end{proof}

\begin{proposition}\label{prop:quasi-m2}
Let $\cG_x$ be a balanced quasi-arc of weight $x$ and let $I\subseteq\{1,2,3\}$ with $|I|=2$.  Then
\[
\alpha_I(s)=
\begin{cases}
0, & s=1,\\[1mm]
\binom{x+2}{2}, & s=2,\\[1mm]
\binom{3x+3}{s}-2\binom{x+2}{s}, & 3\le s\le x+2,\\[1mm]
\binom{3x+3}{s}, & x+2<s\le 3x+2.
\end{cases}
\]
Moreover,
\begin{equation}\label{eq:quasi-m2}
\E[\tau_I(\cG_x)]
=3+\frac{2}{3x+1}+\frac{1}{3x+2}
-\frac{(x+2)(x+1)}{(3x+2)(3x+1)}
+2\sum_{s=3}^{x+2}\prod_{j=0}^{s-1}\frac{x+2-j}{3x+2-j}.
\end{equation}
\end{proposition}

\begin{proof}
Assume, without loss of generality, that $I=\{1,2\}$.  Two points recover $E_1$ and $E_2$ exactly when they lie on the line $E_1E_2$, which contains $x+2$ points of $\cG_x$.  For $3\le s\le x+2$, the non-recovering $s$-subsets are precisely those contained in one of the two other fundamental lines.  For $s>x+2$, no $s$-subset can be contained in a fundamental line, so every subset works.  The expectation follows from Lemma~\ref{lem:general-expectation}.
\end{proof}

\begin{proposition}\label{prop:quasi-m3}
Let $\cG_x$ be a balanced quasi-arc of weight $x$ and let $I=\{1,2,3\}$.  Then
\[
\alpha_I(s)=
\begin{cases}
0, & s=1,2,\\[1mm]
\binom{3x+3}{s}-3\binom{x+2}{s}, & 3\le s\le x+2,\\[1mm]
\binom{3x+3}{s}, & x+2<s\le 3x+2.
\end{cases}
\]
Consequently,
\begin{equation}\label{eq:quasi-m3}
\E[\tau_I(\cG_x)]
=3+\frac{2}{3x+1}+\frac{1}{3x+2}
+3\sum_{s=3}^{x+2}\prod_{j=0}^{s-1}\frac{x+2-j}{3x+2-j}.
\end{equation}
\end{proposition}

\begin{proof}
For $s=1,2$ it is impossible to generate the whole plane.  For $3\le s\le x+2$, the only bad subsets are those contained in one of the three fundamental lines.  For larger $s$ this cannot happen, and every subset contains three non-collinear points.  The expectation is again obtained from Lemma~\ref{lem:general-expectation}.
\end{proof}

The three formulae have different limiting behavior.

\begin{corollary}\label{cor:quasi-asymptotics}
For balanced quasi-arcs of weight $x$,
\[
\lim_{x\to\infty}\E[\tau_I(\cG_x)]
=
\begin{cases}
\frac{17}{6}, & |I|=1,\\[1mm]
3, & |I|=2,\\[1mm]
\frac{19}{6}, & |I|=3.
\end{cases}
\]
\end{corollary}

\begin{proof}
For fixed $s$,
\[
\prod_{j=0}^{s-1}\frac{x+2-j}{3x+2-j}\longrightarrow \left(\frac13\right)^s.
\]
The products are dominated by a geometric sequence, so the sums converge to
\[
\sum_{s=3}^{\infty}\left(\frac13\right)^s=\frac{1}{18}.
\]
Substituting this limit in \eqref{eq:quasi-m1}, \eqref{eq:quasi-m2}, and \eqref{eq:quasi-m3} gives the three stated values.
\end{proof}

\section{Discussion, comparisons, and open directions}\label{sec:discussion}

{We developed a cardinality-based extremal and finite-geometric study of generalized random access, interpolating between the recovery of one information symbol and the recovery of the full information vector.  Starting from the known subset-counting formula \eqref{eq:general-expectation}, we studied the expectation uniformly over all requested sets of a fixed cardinality and made its dependence on the geometry of the projective system explicit.}

The formulae above show why the generalized model is useful for understanding the extremal cases.  By varying $|I|$, one sees how the same projective system behaves as the recovery task moves from local to global.  In particular, there is no single universally best geometry for all random access tasks; the best behavior depends on the size of the requested set $I$.

In dimension three, the three relevant regimes are particularly transparent.  For $|I|=1$, systematic MDS encoders and simplex encoders have expectation $3$, while balanced quasi-arcs have limiting expectation $17/6<3$.  This recovers and explains the advantage of balanced quasi-arcs for single-symbol random access.  For $|I|=2$, the balanced quasi-arc expectation tends to $3$, just as the MDS value does when the length tends to infinity.  For finite, length-matched examples, the quasi-arc values can be slightly below the MDS values.  For $|I|=3$, the MDS value is the smallest one, in agreement with the optimality of systematic MDS encoders for full recovery and with the coverage-depth point of view of \cite{bar2024cover,bertuzzo2025coverage}.

\Cref{tab:numerics} reports representative length-matched comparisons in dimension three.  In the first two columns, $q$ denotes the field size of the simplex encoder and $n=q^2+q+1$ is its length.  For the quasi-arc column, we choose the weight $x$ so that $3x+3=n$, whenever such a balanced quasi-arc exists over a suitable field.  The MDS column gives the value of a systematic MDS encoder of the same length, when such an encoder exists over the relevant field; otherwise, it should be interpreted as the corresponding MDS benchmark value.  The smallest value in each row is highlighted.

\begin{table}[ht]
\centering
\small
\begin{tabular}{cccccc}
\toprule
$q$ & $n$ & $|I|$ & Simplex & Balanced quasi-arc & MDS \\
\midrule
4 & 21  & 1 & 3.0000 & \textbf{2.8470} & 3.0000 \\
4 & 21  & 2 & 3.3000 & \textbf{3.1439} & 3.1500 \\
4 & 21  & 3 & 3.3625 & 3.3593 & \textbf{3.1553} \\
\midrule
7 & 57  & 1 & 3.0000 & \textbf{2.8379} & 3.0000 \\
7 & 57  & 2 & 3.1607 & \textbf{3.0509} & 3.0536 \\
7 & 57  & 3 & 3.1811 & 3.2343 & \textbf{3.0542} \\
\midrule
13 & 183 & 1 & 3.0000 & \textbf{2.8347} & 3.0000 \\
13 & 183 & 2 & 3.0824 & \textbf{3.0156} & 3.0165 \\
13 & 183 & 3 & 3.0883 & 3.1873 & \textbf{3.0165} \\
\bottomrule
\end{tabular}
\caption{Expected generalized random access variables in dimension $3$.  Here $q$ is the simplex field size and $n=q^2+q+1$; the quasi-arc weight is chosen by $3x+3=n$.}
\label{tab:numerics}
\end{table}

These comparisons suggest the following interpretation.  Balanced quasi-arcs deliberately concentrate extra collinear structure on the three fundamental lines.  This makes the fundamental information points easier to recover individually and in pairs, because many small subsets already generate the requested point or the requested fundamental line.  The same concentration becomes a disadvantage for full recovery: collinear subsets on a fundamental line delay the generation of the whole plane.  By contrast, systematic MDS encoders do not privilege any small subset of information coordinates, but they guarantee that every set of $k$ columns generates the whole ambient space, which is exactly the desired property for recovering all information symbols.

The same phenomenon appears from the dual-code viewpoint.  Low-weight dual codewords through prescribed systematic positions create alternative recovering sets and can therefore improve partial recovery.  However, a large supply of such dependencies may be incompatible with the strongest full-recovery behavior.  Thus the generalized random access problem provides a way to measure the trade-off between local recovery advantages and global recovery efficiency.

Several problems remain open.  First, it would be useful to determine whether the quasi-arc advantage for $|I|=2$ can be proved uniformly for natural length-matched families, rather than only observed through formulae and numerical comparisons.  Second, one could extend the construction of geometrically biased projective systems beyond dimension three, where the role of fundamental lines should be replaced by higher-dimensional flats.  Third, it would be interesting to characterize the dual codewords that are most useful for generalized random access: the dual viewpoint suggests that families with many low-weight dual codewords through prescribed systematic positions can improve partial recovery, while this may conflict with good performance for full recovery.  Finally, the generalized random access model could be studied together with other storage constraints, such as error correction, locality, or non-uniform sampling distributions.

\section*{Acknowledgments}
{While completing this manuscript, the authors became aware of the recent preprint by Bar-Lev \cite{BarLev2026CodedRetrieval}, which studies a related two-file retrieval problem and contains results overlapping with parts of the present work, including the systematic-MDS calculation.  The overlapping systematic-MDS results in the present paper were obtained independently and were included in the Master’s thesis of Antonio Petrillo, defended in 2025, before the authors became aware of Bar-Lev’s work.}

This work was supported by a research grant (VIL52303) from Villum Fonden.  The third author is grateful for the hospitality of the Algebra group at DTU during the development of this research in August 2025.  The research of the third author was partially supported by the Italian National Group for Algebraic and Geometric Structures and their Applications (GNSAGA - INdAM).

\end{document}